\documentclass[10pt,english]{article}
\usepackage{graphicx} 
\usepackage[T1]{fontenc}
\usepackage[margin=2cm]{geometry}
\usepackage{amssymb,amsmath,amsthm, amsfonts}
\usepackage[ruled]{algorithm2e}
\usepackage{hyperref}

\newtheorem{theorem}{Theorem}
\newtheorem{lemma}{Lemma}

\newtheorem{definition}{Definition}

\usepackage{biblatex} 
\author{
  Susanne Albers \thanks{\texttt{susanne.albers@tum.de}.Technical University of Munich, Germany; TUM School of Computation, Information and Technology, Department of Computer Science}
  \and
  G. Wessel van der Heijden \thanks{\texttt{wessel.heijden@tum.de}. Technical University of Munich, Germany; TUM School of Computation, Information and Technology, Department of Computer Science. Work supported by the Deutsche Forschungsgemeinschaft (DFG, German Research Foundation) - GRK 2201/2 - Projektnummer 277991500}
}

\title{Scheduling to Maximize Weighted Throughput with an Active-Time Budget}

\date{June 2026}
\date{}

\begin{document}

\maketitle

\begin{abstract}
We study the active-time scheduling problem with weighted throughput maximization. 
In this setting, a set of $n$ jobs $J$ arrive at integer release times, each with an integer processing time and integer deadline. 
Jobs may be preempted at integer time slot boundaries. 
A schedule assigns jobs to time slots, with at most $m$ jobs assigned to the same time slot. 
A slot is called \emph{active} if at least one job is scheduled in it. 
Instead of scheduling all jobs to minimize the number of active time slots, we consider the more general variant of \emph{weighted throughput} with an active-time budget $K$, where each job $j\in J$ has a weight $w_j$. 
The objective is to maximize the total weight of \emph{completed} jobs using at most $K$ active time slots. 
This means that partially scheduled jobs do not count towards the objective. 
The classical active-time minimization problem is recovered by asking whether all jobs can be completed within a given active-time budget. 
We give hardness, approximation, and exact algorithmic results. 
For general intervals with unbounded parallelism, we prove NP-hardness, rule out an FPTAS unless $\mathrm{P}=\mathrm{NP}$, and give a pseudo-polynomial time $\Omega(1/\log K)$-approximation. 
For proper intervals, we prove a canonical structural lemma and obtain an exact $(nK)^{O(m)}$-time algorithm.
For laminar intervals, we give an exact $f(K,m)\cdot n^{O(1)}$-time algorithm. 
\end{abstract}

\section{Introduction}
Scheduling for parallel machines is a fundamental problem in computer science and operations research. 
Energy consumption is an important objective in scheduling and has been studied extensively \cite{DBLP:journals/cacm/Albers10}. 
In energy sensitive environments such as data centers and cloud platforms, a basic way to reduce cost is to reduce the amount of time that a machine is active. 
Some models penalize processor activation with explicit startup costs \cite{DBLP:journals/scheduling/DemaineGHSZ13}, which gives algorithms that avoid idle gaps. 
In contrast, in cloud-computing settings the cost of a rented machine is often closely related to the active usage time. 
This naturally leads to scheduling problems with an \emph{active-time budget}. 
Given a limited amount of machine active-time, complete a subset of the jobs. 
We study this question in the active-time scheduling model introduced by Chang et al. \cite{DBLP:journals/algorithmica/ChangGK14}. 
Unlike the classical objective of completing all jobs while minimizing active time, we are given an active-time budget and the objective is to maximize the total weight of completed jobs. 

We consider the following scheduling problem. 
We are given a set of \(n\) jobs \(J=\{1,\dots,n\}\). 
Each job \(j\in J\) has an integer release time \(r_j\), an integer deadline \(d_j\), an integer processing time \(p_j\), and a nonnegative weight \(w_j\). 
Time is partitioned into discrete integer time slots and jobs may be preempted at time slot boundaries. 
In each slot, at most \(m\) jobs can be processed simultaneously. 
We say a time slot \(t\) is \emph{active} when at least one job unit is assigned to time slot \(t\). 
The schedule may use at most \(K\) active time slots. 
A job \(j\) is completed if it receives \(p_j\) units of processing in time slots contained in its window \(W_j=[r_j,d_j]\). 
The objective is the total weight of completed jobs, also called the weighted throughput. 
If a job is only partially scheduled, it does not count towards the objective. 
Any processing assigned to an uncompleted job can be deleted without hurting the objective or violating feasibility. 
Throughout the paper, we assume w.l.o.g. that every processed job is completed. 
As in much of the active-time scheduling literature, we use the discrete-time model and state running times that are polynomial in numeric parameters such as \(K\) and the time horizon. 
These running times are pseudo-polynomial under binary encoding and are comparable to the standard time-expanded flow formulations used to check feasibility for active-time minimization \cite{DBLP:journals/scheduling/ChangKM17,DBLP:conf/spaa/KumarK18,DBLP:conf/isaac/CaoFLMRU22,DBLP:journals/scheduling/CalinescuW21}. 

Scheduling with an active-time budget to maximize weighted throughput contains the feasibility version of the active-time scheduling model introduced by Chang et al. \cite{DBLP:journals/algorithmica/ChangGK14}. 
In the original active-time minimization problem, every job must be scheduled for \(p_j\) time slots within its window \([r_j,d_j]\) and the goal is to minimize the total active time slots. 
Our model looks at the opposite perspective, where the active-time budget \(K\) is fixed, and the algorithm must decide which jobs to complete. 
The minimization problem can be recovered by asking, for a given value of \(K\), whether the optimum throughput equals the total weight of all jobs. 
A binary search over \(K\) then gives the minimum active-time. 
For \(p_j=1\), Chang et al. \cite{DBLP:journals/algorithmica/ChangGK14} gave an \(O(n\log n)\)-time algorithm that solves the active-time minimization problem optimally. 
They also considered throughput maximization for unit length jobs and an active-time budget \(K\le n\), for which they obtained an optimal dynamic programming algorithm with a (pseudo-polynomial) running time of \(O(n^5K)\). 
For arbitrary processing times, there have been several attempts to obtain good approximation algorithms for the active-time minimization problem. 
Chang et al. \cite{DBLP:journals/scheduling/ChangKM17} gave a \(2\)-approximation for the active-time scheduling problem with arbitrary processing times using an LP rounding procedure. 
Kumar and Khuller \cite{DBLP:conf/spaa/KumarK18} showed that a \(2\)-approximate ratio can also be obtained using a simple greedy procedure for the active-time scheduling problem that requires calculating feasibility over all jobs at each iteration. 
Later, C\u{a}linescu and Wang \cite{DBLP:journals/scheduling/CalinescuW21} gave an alternative, more involved LP rounding procedure achieving the same \(2\)-approximation ratio. 
Its integrality gap is at least \(5/3\), but they conjecture it is strictly better than \(2\). 
Cao et al. \cite{DBLP:conf/isaac/CaoFLMRU22} broke the \(2\)-approximation barrier for the nested active-time scheduling setting, giving a \(9/5\)-approximation using a complex LP rounding procedure. 
The complexity of the active-time scheduling problem remained open for a while, until Saha and Purohit \cite{DBLP:journals/corr/abs-2112-03255} showed that active-time scheduling problem with arbitrary processing times is NP-hard, via a reduction from Balanced SAT. 
Later Cao et al. showed that the active-time scheduling problem remains NP-hard, even with nested intervals \cite{DBLP:conf/isaac/CaoFLMRU22}. 

In a survey on busy-time and active-time scheduling, Chau and Li raised the throughput maximization setting as an interesting open research direction and as a way to settle the complexity status of the active-time scheduling problem \cite{DBLP:conf/birthday/ChauL20}. 
To the best of our knowledge, the only prior work on throughput maximization in the active-time model was by Chang et al. \cite{DBLP:journals/algorithmica/ChangGK14}, who considered unit processing time jobs. 
Therefore, the weighted-throughput problem with arbitrary processing times has remained largely unexplored.

\subsection{Related work}
Several scheduling models are closely related to the active-time scheduling problem. 
We briefly discuss the most relevant ones. 

\subsubsection{Busy-time scheduling}
Busy-time scheduling is closely related to active-time scheduling. 
It can be seen as the non-preemptive version of active-time scheduling. 
Jobs are assigned to machines with the objective to minimize the total time during which machines are busy. 
This model is motivated by energy efficiency and was originally studied in connection with network design \cite{DBLP:conf/soda/WinklerZ03}. 
Unlike active-time scheduling, busy-time scheduling is non-preemptive and allows several machines, each with a bounded parallelism parameter. 
This difference rules out constant competitive ratios for bounded parallelism in some online settings \cite{DBLP:journals/tcs/ShalomVWYZ14} and results in different solution techniques. 
For a broader comparison of busy-time and active-time scheduling, we refer to the survey of Chau and Li \cite{DBLP:conf/birthday/ChauL20}.

\subsubsection{Precedence constrained scheduling}
Papadimitriou and Yannakakis studied scheduling unit-time jobs with interval-ordered precedence constraints on \(m\) processors and observed that this model can be interpreted as minimizing the number of time steps in which at least one processor is active \cite{DBLP:journals/siamcomp/PapadimitriouY79}. 
They obtained a greedy \(O(n^2)\) algorithm which was improved by Chang et al. to \(O(n\log n)\) \cite{DBLP:journals/algorithmica/ChangGK14} for the active-time scheduling problem. 
\footnote{In Appendix~\ref{sec:hardness}, we describe a related relaxation of active-time scheduling, which we call \emph{Interval-Ordered Chains} (IOC).
We believe this setting is of independent interest to the precedence-constrained scheduling community. }

\subsubsection{Calibration scheduling}
Calibration scheduling is another related activation model where a machine can be calibrated for \(B\) consecutive time slots, during which it can process one unit-sized job per slot. 
The objective is typically to minimize the number of calibrations. 
The weighted throughput maximization version was considered by Chau et al. \cite{DBLP:conf/wads/ChauFLWZ019}. 
For unit jobs, they obtained a \((1/3)\)-approximation and for arbitrary processing times they obtained a \(((1-\varepsilon)/3\)-approximation in pseudo-polynomial time and a \(((1-\varepsilon)/18\)-approximation in polynomial time.

\subsection{Our contributions}
We study the problem of scheduling jobs to maximize weighted throughput with an active-time budget. 
To the best of our knowledge, this is the first systematic study of the weighted throughput variant with arbitrary processing times. 
Our results are summarized below. 
\begin{center}
\begin{tabular}{lll}
\hline
Interval family & Parameters & Result \\
\hline
General & \(m=\infty\) & NP-hard; no FPTAS unless \(\mathrm{P}=\mathrm{NP}\) \\
& & \(\Omega(1/\log K)\)-approx. in \(\mathrm{poly}(n,K)\) time \\
Proper & parameterized \(m\) & exact \((nK)^{O(m)}\) XP \\
Laminar & parameterized \(K,m\) & exact \(f(K,m)\cdot n^{O(1)}\) FPT \\
\hline
\end{tabular}
\end{center}

We first consider the hardness of the setting with unbounded \(m\), and show this is already hard. 
With \(m=\infty\), the problem becomes to choose active time slots within interval demands. 
We prove NP-hardness and give an \(\Omega(1/\log K)\)-approximation by grouping jobs by processing time. 
For proper intervals, we prove a canonical-ordering lemma and show that there exists an optimum schedule with at most \(2m-1\) jobs split across any time-slot boundary. 
This results in an exact \((nK)^{O(m)}\)-time algorithm, which is pseudo-polynomial XP in the capacity parameter \(m\). 
For laminar intervals, we use the nested structure in a dynamic program over capacity profiles, which gives an exact FPT algorithm parameterized by \(K\) and \(m\). 

The remainder of the paper is organized as follows. 
Section~\ref{sec:general-unbounded} considers general intervals with unbounded parallelism and proves the hardness and approximation results. 
Section~\ref{sec:proper} considers proper intervals and gives the XP algorithm parameterized by \(m\). 
Section~\ref{sec:nested} gives the FPT algorithm for laminar intervals.

\section{General intervals with unbounded parallelism}\label{sec:general-unbounded}
In this section, we consider general intervals with unbounded parallelism. 
When parallelism is unbounded, capacity is no longer an issue since an active time slot can process one unit of every job with a window that contains that slot. 
For active-time minimization, this setting already admits an optimal algorithm using a greedy activation rule \cite{DBLP:journals/scheduling/ChangKM17}. 
This does not hold for the weighted throughput setting, since jobs must be selected before time slots are committed and an accepted job can force additional activations. 
A simple instance in which all job windows are disjoint and \(m\) is unbounded already captures the knapsack problem. 
The knapsack problem can be written as 
\(
\max\left\{\sum_{i=1}^{n}w_ix_i : \sum_{i=1}^{n}c_ix_i\le K\right\}
\) 
with \(x_i\in\{0,1\}\) and where \(w_i\) is the profit of item \(i\), \(c_i\) is its cost, and \(K\) is the budget. 
We recover this problem in the active-time formulation as follows. 
For each item \(i\), create one job \(j\) with weight \(w_j=w_i\) and processing time \(p_j = c_i\). 
Choose the windows \([r_j,d_j]\) to be pairwise disjoint and long enough to schedule the corresponding job. 
Accepting a job \(j\) for disjoint jobs costs exactly \(p_j\) to the active-time budget and adds \(w_j\) weight to the objective. 
Therefore, this models the knapsack problem. 
Note that the knapsack problem is weakly NP-hard, and a pseudo-polynomial DP exists. 
When windows overlap, the cost to the budget of accepting a job is no longer fixed and depends on which active time slots were opened for other accepted jobs. 
The marginal active-time cost of job \(j\) is \(p_j\) minus the number of already active slots in \([r_j,d_j]\) that can be used by \(j\). 
This means that the marginal cost depends on the set of selected jobs and on the locations of the active slots. 
This is the main distinction from 0/1 knapsack. 
The cost of a selected job set is not described by a fixed item cost, but it depends on where the active time slots are placed and on how different jobs share these time slots. 

\begin{theorem}
    Weighted throughput active-time scheduling on general intervals is NP-hard even when \(m=\infty\). 
    Furthermore, no FPTAS exists for the weighted throughput scheduling problem with an active-time budget unless \(\mathrm{P}=\mathrm{NP}\). 
\end{theorem}
\begin{proof}
We consider the hardness of the problem using a reduction from the decision version of the maximum directed cut on acyclic digraphs. 
The cardinality maximum directed cut (\textsc{MaxDiCut}) on DAGs (directed acyclic graphs) was proven to be NP-hard and APX-hard by Lampis et al. \cite{DBLP:journals/disopt/LampisKM11}. 
We define the cardinality \textsc{MaxDiCut} problem on DAGs. 
Given is a directed acyclic graph \(G=(V,E)\) with vertex set \(V=\{1,\dots,n\}\) (topologically ordered) and directed arc set \(E\) and an integer \(B\). 
The objective is to decide whether there is a set \(S\subseteq V\) such that at least \(B\) arcs go from \(S\) to \(V\setminus S\). 
We reduce this problem to weighted-throughput active-time scheduling with unbounded \(m=\infty\). 
For each \(i\in V\), let \(a_i = 2i-1\) and \(b_i=2i\) be two adjacent time slots. 
Intuitively, activating time slot \(a_i\) indicates that \(i\notin S\), and choosing \(b_i\) indicates \(i\in S\). 
The active-time budget is \(K=n\). 
We consider two types of jobs. 
One set of jobs are the vertex jobs \(i\) with job window \(W_i=[a_i,b_i]\), processing time \(p_i=1\) and weight \(w_i=M\) for a large \(M= n^2+1\). 
The second set of jobs are the arc jobs \((i,j)\), one for each \(i,j\in E\) with \(i<j\) by the topological ordering with job window \(W_{ij}=[b_i,a_j]\), processing time \(p_{ij} = j - i + 1\) and weight \(w_{ij}=1\). 
See Figure~\ref{fig:general_unbounded_np-hardness} for an example construction. 

Any optimal solution has to schedule all vertex jobs. 
When one of \(a_i\) and \(b_i\) is activated for each \(i\in V\), this gives a feasible schedule with an objective weight of at least \(nM\). 
Any schedule that is missing one vertex job has an objective value of at most \((n-1)M + |E| < nM\) since \(M=n^2+1>n^2\ge|E|\). 
The pairs \(\{a_i,b_i\}\) are pairwise disjoint and each vertex job needs one slot inside its own pair, so scheduling all \(n\) vertex jobs uses all \(K=n\) active slots. 
Since an optimal schedule must activate exactly one of \(a_i,b_i\) for each \(i\), no other slot is active. 
This means that the remaining choice is which of those two time slots \(a_i,b_i\) to activate for each vertex \(i\). 
We interpret \(b_i\) active as \(i\in S\), and \(a_i\) active as \(i\notin S\). 

Consider an arc job \((i,j)\in E\). 
Its window is \(W_{ij}=[b_i,a_j]\), and one time slot is active for each intermediate vertex \(i+1,\ldots,j-1\). 
Therefore, for any active set of time slots that schedules all vertex jobs, the set of active time slots in its window is \(|A\cap W_{ij}| = (j-i-1) + |A\cap\{b_i\}| + |A\cap\{a_j\}|\). 
Since \(p_{ij}=j-i+1\), the arc job \((i,j)\) is completed if and only if both \(b_i\) and \(a_j\) are active. 
Similarly, it is completed if and only if \(i\in S\) and \(j\notin S\). 
After all vertex jobs are accepted, the additional weight that is gained for the objective from arc jobs is exactly the number of arcs crossing from \(S\) to \(V\setminus S\). 
The objective value for the weighted throughput scheduling problem becomes \(nM+\max_{S\subseteq V}\delta^+(S)\), where \(\delta^+(S)\) denotes the number of outgoing arcs from \(S\) to \(V\setminus S\). 
Therefore, the original \textsc{MaxDiCut} instance has a directed cut of size at least \(B\) if and only if the constructed scheduling instance has value at least \(nM+B\). 
This establishes NP-hardness. 

The same reduction also rules out an FPTAS unless \(\mathrm{P}=\mathrm{NP}\). 
All weights are integers, so every solution value is integer and the optimum is at most \(U=nM+|E|=O(n^3)\). 
Running an FPTAS with \(\varepsilon=\tfrac{1}{U+1}\) takes polynomial time in the input size and \(1/\varepsilon\), i.e. polynomial, and returns a value greater than \((1-\varepsilon)\mathrm{OPT}>\mathrm{OPT}-1\). 
Since the optimum is integer, this equals \(\mathrm{OPT}\) and solve the NP-hard decision problem. 
\end{proof}

\begin{figure}
    \centering
    \includegraphics[page=1]{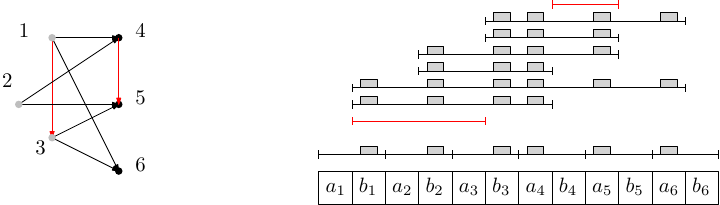}
    \caption{Example reduction from the maximum directed cut on acyclic digraph to the maximum weight throughput problem with an active-time budget. The gray rectangles indicate where a job unit is scheduled. The red intervals and arcs indicate the arcs or jobs that are not in the cut or not completed. }
    \label{fig:general_unbounded_np-hardness}
\end{figure}

We show an approximation algorithm for the problem. 
When \(m=\infty\), capacity constraints disappear. 
Therefore, a job \(j\) is completed exactly when the chosen active set \(A\) has at least \(p_j\) active slots in the window \(W_j=[r_j,d_j]\), i.e. a job is active when \(|A\cap W_j|\ge p_j\). 

Throughout this subsection, time is discrete. 
A window \(W_j=[r_j,d_j]\) denotes the set of time slots \(\{r_j,r_j+1,\dots,d_j\}\) and hence \(|W_j|=d_j-r_j+1\). 
A set of active slots \(A\) completes job \(j\) if \(|A\cap W_j|\ge p_j\). 
We first show how to compress the time horizon to a smaller sized set of time slots. 
It is enough to consider only \(K\) slots around release-time and deadline breakpoints. 
This gives a candidate set of \(O(nK)\) time slots, so the algorithms based on this set are polynomial in \(n\) and the numeric value of \(K\). 

\begin{lemma}\label{lem:canonical_time_slots}
    Let \(C=\bigcup_{j\in J}\{r_j-K,\dots,r_j+K\}\cup\{d_j-K,\dots,d_j+K\}\) be the time slots \(K\) around release time and deadline breakpoints. 
    For every active set \(A\) with \(|A|\le K\), there is an active set \(A'\subseteq C\) with \(|A'|\le|A|\) such that every job completed by \(A\) is also completed by \(A'\). 
\end{lemma}
\begin{proof}
    Sort by all release-time and deadline endpoints. 
    Endpoint slots themselves belong to \(C\). 
    Between two consecutive endpoints \(e<f\), every slot in the open gap \(\{e+1,\ldots,f-1\}\) belongs to exactly the same set of job windows, since no release time or deadline occurs inside the gap. 
    If \(A\) uses \(h\) time slots in this gap, then \(h\le|A|\le K\) and the gap has at least \(h\) slots.
    The slots \(e+1,\dots,e+h\) all lie in the gap and since \(h\le k\), they all lie in \(C\). 
    Replace the \(h\) active time slots of \(A\) in the gap by any \(h\) candidate slots in the same gap. 
    This preserves for every job the number of active time slots it overlaps inside the gap. 
    Slots outside the range of all endpoints are contained in no job window and can be deleted. 
    Applying this replacement independently to all gaps gives the desired set \(A'\subseteq C\). 
    After this replacement, any algorithm only has to consider \(O(nK)\) candidate time slots. 
\end{proof}

We give an \(\Omega(1/\log K)\)-approximation for unbounded \(m\). 
Any job with \(p_j>K\) or \(p_j>|W_j|\) can not be scheduled and is discarded, so throughout this subsection we assume \(p_j\le \min\{K, |W_j|\}\) for all jobs \(j\in J\). 
We consider an optimal set of time slots \(A^*\) of value \(\mathrm{OPT}\) and write \(w(S)=\sum_{j\in S}w_j\) for a selected set of jobs \(S\subseteq J\). 
The approximation algorithm partitions the jobs \(J\) into buckets of similar processing time. 
Each bucket can then be solved within a constant factor and the best schedule is kept over all buckets. 
Let \(q_{\max}=\lfloor K/2\rfloor\) and let \(Q=\{2^i:2^i\le q_{\max}\}\cup\{q_{\max}\}\) be the set of bucket sizes. 
Every job with \(p_j\le q_{\max}\) is put in the bucket \(J_q\) with \(q=\min\{q'\in Q: p_j\le q'\}\).
For some \(q\) with a predecessor \(q'\) in \(Q\), then \(q'>q/2\), so for every \(j\in J_q\), we have \(q/2<p_j\le q\). 
Let \(J_{\mathrm{large}}=\{j: p_j>q_{\max}\}\) be the bucket with large jobs. 
Let \(\mathrm{OPT}_q\) and \(\mathrm{OPT}_{\mathrm{large}}\) be the weights of completed jobs in the optimal schedule \(A^*\) corresponding to each bucket. 
Therefore, 
\[
\mathrm{OPT} = \mathrm{OPT}_{\mathrm{large}} + \sum_{q\le K/2}\mathrm{OPT}_q. 
\]

In order to complete a job \(j\) in bucket \(q\), we consider some time slot \(t\in [r_j,d_j]\). 
For any such time slot \(t\), activating all time slots \(\{t-q+1,\dots,t+q\}\) allows job \(j\) to schedule while activating \(2q\) time slots. 
For such a center \(t\), let \(B_q(t) = \{t-q+1,\dots,t+q\}\) be the block of activated time slots of size \(2q\). 
Since \(p_j\le q\), if \(t\in W_j\), then \(B_q(t)\cap W_j\) contains enough slots to complete \(j\). 

\begin{lemma}\label{lem:general_unbounded_schedule_jobs}
    If \(j\in J_q\) and \(t\in W_j\), then \(|B_q(t)\cap W_j|\ge p_j\). 
\end{lemma}
\begin{proof}
    Consider the window \(W_j=[r_j,d_j]\) and recall \(r_j\le t\le d_j\). 
    If \(|W_j|\le q\), then \(r_j\ge d_j-q+1\ge t-q+1\) and \(d_j\le r_j+q-1\le t+q-1\). 
    Therefore \(W_j\subseteq B_q(t)\) and \(|B_q(t)\cap W_j|=|W_j|\ge p_j\). 
    If \(|W_j|>q\), then \(|B_q(t)\cap W_j|=\min\{q,t-r_j+1\}+\min\{q,d_j-t\}\) and therefore either \(|B_q(t)\cap W_j|=|W_j|\ge p_j\) when neither minimum is \(q\), or \(|B_q(t)\cap W_j|\ge q\). 
    In both cases, \(|B_q(t)\cap W_j|\ge\min\{|W_j|,q\}\ge p_j\) since \(p_j\le q\) and \(p_j\le|W_j|\). 
\end{proof}

We now show how to choose these center slots effectively. 
Within budget \(K\), we can select at most \(L_q=\lfloor\tfrac{K}{2q}\rfloor\) centers for bucket \(q\). 
The algorithm greedily selects \(L_q\) centers. 
In each iteration, it chooses an integer time slot \(t\) that maximizes the total weight of currently unhit jobs \(j\in J_q\) with \(t\in W_j\). 
This maximum-weight center time slot can be found by sweeping the release and deadline endpoints of the unhit job intervals. 
Overlapping blocks are allowed, since the active set returned for bucket \(q\) is the union of the chosen blocks. 

\begin{lemma}
    The greedy algorithm for bucket \(q\) returns a feasible set of active time slots \(A_q\) with weight at least \((1 - e^{-1/8})\mathrm{OPT}_q\). 
\end{lemma}
\begin{proof}
    The algorithm chooses at most \(L_q=\lfloor \frac{K}{2q}\rfloor\) centers and each chosen center opens a block of size \(2q\). 
    Therefore \(|A_q|\le 2q\cdot L_q\le K\). 
    Any job hit by a chosen center is completed by Lemma~\ref{lem:general_unbounded_schedule_jobs}. 

    Let \(S^*_q\subseteq J_q\) be the set of jobs from bucket \(q\) that are completed in the optimal schedule \(A^*\). 
    At any iteration, let \(U\subseteq S^*_q\) be the set of unhit jobs, i.e. jobs whose window contains none of the centers selected so far. 
    Each job \(j\in U\) has \(|A^*\cap W_j|\ge p_j > q/2\) since they are part of the optimal schedule, so the total weight \(\sum_{j\in U}w_j|A^*\cap W_j|\ge q/2\cdot w(U)\). 
    Now \(\sum_{j\in U}w_j\sum_{t\in A^*: t\in W_j}1 = \sum_{t\in A^*}\sum_{j\in U: t\in W_j}w_j\). 
    Since any optimal schedule activates at most \(K\) time slots, this sum is distributed over at most \(K\) time slots. 
    Hence there exists a time slot \(t\in A^*\) whose hit weight is at least the average
    \[
    \max_{t\in A^*} \sum_{j\in U: t\in W_j}w_j \ge \frac{q}{2K}\cdot w(U). 
    \]
    Therefore, there exists a center \(t\in A^*\) that hits at least \(\tfrac{q}{2K}\)-fraction of the remaining optimal weight. 
    Let \(\mathrm{ALG}_i\) be the total weight of jobs of \(J_q\) hit by the first \(i\) centers and let \(U_i\subseteq S^*_q\) be the jobs of \(S^*_q\) still unhit after \(i\) centers. 
    Every job of \(S^*_q\setminus U_i\) has been hit and is therefore counted in \(\mathrm{ALG}_i\), so \(w(U_i)\ge\mathrm{OPT}_q-\mathrm{ALG}_i\). 
    By the averaging bound for \(U_i\) and since the algorithm maximizes the weight of newly hit jobs, \(\mathrm{ALG}_{i+1}-\mathrm{ALG}_i\ge \tfrac{q}{2K}w(U_i)\ge\tfrac{q}{2K}(\mathrm{OPT}_q-\mathrm{ALG}_i)\). 
    Therefore \(\mathrm{OPT}_q-\mathrm{ALG}_{i+1}\le(1-\tfrac{q}{2K})(\mathrm{OPT}_q-\mathrm{ALG}_i)\). 
    Repeating this for all \(L_q=\lfloor\tfrac{K}{2q}\rfloor\ge\tfrac{K}{4q}\) centers gives a total hit weight of at least
    \[
    \left(1 - (1 - \frac{q}{2K})^{L_q}\right)\mathrm{OPT}_q
    \ge \left(1 - (e^{- \frac{q}{2K}})^{\frac{K}{4q}}\right)\mathrm{OPT}_q 
    = \left(1 - e^{- 1/8}\right)\mathrm{OPT}_q .
    \]
\end{proof}

This gives a constant approximation within each bucket of \(\alpha\cdot\mathrm{OPT}_q\) where \(\alpha = 1 - e^{-1/8}\). 
We now consider \(J_{\mathrm{large}}\) where each job has \(p_j>K/2\). 
The block greedy argument does not apply directly since \(L_{\mathrm{large}} = \lfloor\tfrac{K}{2p_j}\rfloor=0\) which means no job is selected, so we handle them separately. 

\begin{lemma}\label{lem:general_unbounded_large_contiguous}
    Let \(S\subseteq J_{\mathrm{large}}\) be feasible using at most \(K\) active slots. 
    There exists a contiguous block \(I\) of \(K\) time slots such that \(|I\cap W_j|\ge p_j\) for every \(j\in S\). 
\end{lemma}
\begin{proof}
    Let \(A\) be a feasible active set of time slots for \(S\) with \(|A|\le K\). 
    For every job \(j\in S\), \(|A\cap W_j|>K/2\). 
    So for any two jobs \(i,j\in S\), the active time slots in their window \(|A\cap W_i|+|A\cap W_j|>K\ge |A|\). 
    By inclusion-exclusion, \(|A\cap W_i\cap W_j|\ge|A\cap W_i|+|A\cap W_j|-|A|>0\), so the windows in \(S\) pairwise intersect. 
    Since intervals satisfy the Helly property, there is a slot \(\tau\) contained in every window \(W_j\) for all \(j\in S\). 

    We only consider relevant time slots near the common time slot \(\tau\). 
    Consider an active time slot \(a\in A\), \(a<\tau\), and there is an inactive time slot \(b\notin A\) with \(a<b\le\tau\). 
    Replace \(a\) by \(b\). 
    If \(a\in W_j\), then \(b\in W_j\), because \(W_j\) is an interval containing both \(a\) and \(\tau\), so the replacement cannot decrease \(|A\cap W_j|\) for any \(j\in S\). 
    If \(a\notin W_j\), then the number of active time slots in the window of job \(j\) does not decrease. 
    The same argument applies symmetrically to an active time slot \(a>\tau\). 
    Each replacement strictly decreases the total distance between \(\tau\) and the active time slots \(a\), so the replacements terminate when we have a contiguous block around \(\tau\) of size \(|A|\le K\) that completes every job in \(S\). 
    Extending this block by adjacent slots until it has exactly \(K\) slots keeps it contiguous and every \(j\in S\) can still be scheduled, preserving feasibility. 
\end{proof}

Lastly, we show how to find the best interval \(I\) of size \(K\). 
Let \(I_s=\{s,s+1,\ldots,s+K-1\}\) be an interval of size \(K\) starting at time slot \(s\). 
Since \(p_j\le\min\{K,|W_j|\}\), a job \(j\in J_{\mathrm{large}}\) satisfies \(|I_s\cap W_j|\ge p_j\) exactly when \(s\in F_j:=[r_j+p_j-K, d_j-p_j+1]\), so the task is to find a point stabbing a maximum-weight subset of intervals \(F_j\). 
Some optimal point is a left endpoint of an \(F_j\), so we sweep over these \(n\) endpoints in \(O(n\log n)\) time. 
By Lemma~\ref{lem:general_unbounded_large_contiguous}, the resulting block has weight of at least \(\mathrm{OPT}_{\mathrm{large}}\). 

\begin{theorem}
    For general intervals with unbounded parallelism, there is an algorithm with running time polynomial in \(n\) and the numeric value of \(K\) returning a schedule of value at least 
    \[
    \mathrm{ALG}\ge \frac{1-e^{-1/8}}{\lceil \log K\rceil+1}\mathrm{OPT}.
    \]
\end{theorem}
\begin{proof}
For each regular bucket \(q\), the algorithm finds a schedule of value at least \(\alpha\cdot\mathrm{OPT}_q\) where \(\alpha = 1-e^{-1/8}\). 
For the large bucket, the algorithm finds a schedule of value at least \(\mathrm{OPT}_{\mathrm{large}}\). 
Since there are \(O(\log K)\) buckets in total, the best bucket gives an \(\Omega(1/\log K)\)-fraction of the total value. 
Therefore, there exists some bucket that contains a weight of at least 
\[
\max\Bigl\{\max_q \alpha\,\mathrm{OPT}_q,\ \mathrm{OPT}_{\mathrm{large}}\Bigr\}
\ge \frac{\alpha\,\mathrm{OPT}_{\mathrm{large}}+\alpha\sum_q \mathrm{OPT}_q}{\lceil\log K\rceil+1}
= \frac{\alpha\,\mathrm{OPT}}{\lceil\log K\rceil+1}.
\]
\end{proof}

\section{Proper intervals with parameterized \texorpdfstring{\(m\)}{m}}\label{sec:proper}
Job windows are \emph{proper} when no window strictly contains another, so sorting by nondecreasing release times also sorts the deadlines. 
We index the jobs \(1,\ldots,n\) in this order, breaking ties arbitrarily, so \(i<j\) implies \(r_i\le r_j\) and \(d_i\le d_j\). 
We discard any job with \(p_j>K\), and assume that a schedule processes a job if and only if it completes it. 
We write \(A(t)\subseteq J\) for the set of jobs processed in slot \(t\) with capacity \(|A(t)|\le m\). 

Our main structural lemma is that some optimal schedule splits at most \(2m-1\) jobs across any time slot boundary. 
Scanning time slots form left to right, only \(O(m)\) partially processed jobs have to be remembered explicitly, while other live jobs can be summarized by a single cut index. 
This gives an exact dynamic program running in \((nK)^{O(m)}\) time, which is XP in \(m\) and pseudo-polynomial in \(K\).

\begin{definition}[Inversion pair]\label{def:proper_inversion_pair}
    Consider jobs \(i,j\in J\) with \(i<j\). 
    The pair \((i,j)\) is an \emph{inversion pair} if there are time slots \(t<t'\) such that \(j\in A(t)\) and \(i\notin A(t)\), while \(i\in A(t')\) and \(j\notin A(t')\). 
\end{definition}

\begin{lemma}[Canonical schedule]\label{lem:proper_exchange}
    Every feasible schedule can be transformed, without activating additional time slots and without changing the completed jobs, into a schedule with no inversion pairs. 
    So every proper-interval instance has an optimal schedule with no inversion pairs. 
    We call such a schedule \emph{canonical}. 
\end{lemma}
\begin{proof}
    Let \((i,j)\) with \(i<j\) be an inversion pair with time slots \(t<t'\) where \(j\in A(t)\not\ni i\) and \(i\in A(t')\not\ni j\). 
    Both windows contain both slots, i.e \(r_i\le r_j\le t<t'\le d_i\le d_j\). 
    Swap the unit of \(j\) at \(t\) with the unit of \(i\) at \(t'\). 
    Since \(i\not\in A(t)\) and \(j\not\in A(t')\), neither slot gets a duplicate, and both slots are non-empty. 
    Therefore, capacities, windows, the active time slots and the completed jobs are all unchanged. 
    Considering the inversion pairs from left to right and in increasing index order, this decreases the lexicographic sequence of job indices for each swap and therefore the process terminates. 
\end{proof}

We use the following immediate consequence repeatedly. 

\begin{lemma}\label{lem:proper_canonical_monotonicity}
    Consider a canonical schedule and two jobs \(i,j\) with \(i<j\). 
    If \(j\in A(t)\) and \(i\notin A(t)\) for some slot \(t\), then for every later slot \(t'>t\), \(i\in A(t')\) implies \(j\in A(t')\). 
\end{lemma}
\begin{proof}
    Otherwise \(j\) is present without \(i\) at the earlier time slot \(t\), and \(i\) is present without \(j\) at the later time slot \(t'\). 
    This is exactly an inversion pair \((i,j)\), contradicting canonicality. 
\end{proof}

We next bound the number of jobs that can cross a boundary. 
A job \(j\) is \emph{split} at boundary \(h\) if it is processed both in a slot \(t\le h\) and in a slot \(t\ge h+1\). 
The intuition is that in a canonical schedule, an earlier split job processed after the boundary can force later split jobs to be present in the same post-boundary slot. 
If \(2m\) jobs were split, this would then create a time slot that has to contain \(m+1\) jobs. 

\begin{lemma}\label{lem:proper_few_splits}
    In every canonical schedule, at most \(2m-1\) jobs are split at any boundary \(h\). 
\end{lemma}
\begin{proof}
    Let the split jobs at boundary \(h\) be \(s_1<s_2<\cdots<s_q\). 
    For a pre-boundary slot \(\tau\le h\), let \(P(\tau)=A(\tau)\cap\{s_1,\ldots,s_q\}\) be the split jobs scheduled at \(\tau\), and for a post-boundary slot \(t\ge h+1\) let \(Q(t)=A(t)\cap\{s_1,\ldots,s_q\}\) be the split jobs scheduled at \(t\). 

    First, we show a consequence of canonicality, similar to Lemma~\ref{lem:proper_canonical_monotonicity}. 
    Let \(\tau\le h<t\). 
    If \(a<b\), \(s_a\in Q(t)\setminus P(\tau)\), and \(s_b\in P(\tau)\), then \(s_b\in Q(t)\). 
    Suppose \(s_b\notin Q(t)\).
    Then the earlier time slot \(\tau\) contains \(s_b\) but not \(s_a\), while the later time slot \(t\) contains \(s_a\) but not \(s_b\). 
    This is an inversion pair, so the schedule was not canonical. 
    We call this the \emph{crossing rule}. 

    For a contradiction, suppose that \(q\ge 2m\). 
    We prove by induction on \(i=1,\ldots,m\) the following claim: there exists a post-boundary slot \(t_i\) such that \(s_{m+i}\in Q(t_i)\) and \(|Q(t_i)\cap\{s_1,\ldots,s_{m+i}\}|\ge i+1\). 

    For \(i=1\), choose a pre-boundary slot \(\tau\) that contains \(s_{m+1}\). 
    Time slot \(\tau\) exists since \(s_{m+1}\) is split. 
    Since \(\tau\) has capacity \(m\) and already contains \(s_{m+1}\), at least one job \(s_a\in\{s_1,\ldots,s_{m}\}\) is absent from \(\tau\). 
    Since \(s_a\) is split, we can choose a post-boundary time slot \(t_1\) that contains a job unit of \(s_a\). 
    By applying the crossing rule to \(s_a<s_{m+1}\), it holds that \(s_{m+1}\in Q(t_1)\). 
    Therefore, \(Q(t_1)\) contains both \(s_a\) and \(s_{m+1}\), which proves the base case. 

    Assume that the claim holds for all values up to \(i\), where \(1\le i<m\). 
    Let \(y=s_{m+i+1}\), and choose a pre-boundary time slot \(\tau\) that contains \(y\). 
    Let \(t_i\) be the witness after the boundary from the induction hypothesis, and define \(R=Q(t_i)\cap\{s_1,\ldots,s_{m+i}\}\). 
    Then \(|R|\ge i+1\). 
    Assume some job \(x\in R\) is not in \(P(\tau)\) with \(y\in P(\tau)\). 
    Since \(x<y\), \(x\in Q(t_i)\setminus P(\tau)\), and \(y\in P(\tau)\), the crossing rule implies \(y\in Q(t_i)\). 
    Since \(Q(t_i)\) contains \(y\) and at least the \(i+1\) jobs of \(R\), it contains at least \(i+2\) jobs from \(\{s_1,\ldots,s_{m+i+1}\}\). 
    Therefore, \(t_{i+1}=t_i\) is a valid witness. 

    We now consider the case \(R\subseteq P(\tau)\) in two subcases. 
    \begin{enumerate}
        \item First suppose that \(P(\tau)\) contains all jobs \(s_{m+1},\ldots,s_{m+i+1}\). 
        These are \(i+1\) jobs.
        Since \(\tau\) has capacity \(m\), it cannot also contain \(s_1,\ldots s_m\). 
        Choose a job \(x\in\{s_1,\ldots,s_m\}\setminus P(\tau)\). 
        Since \(x\) is split, choose a post-boundary time slot \(t\) containing \(x\). 
        For every \(z\in\{s_{m+1},\ldots,s_{m+i+1}\}\), we have \(x<z\), \(x\in Q(t)\setminus P(\tau)\) and \(z\in P(\tau)\). 
        By the crossing rule \(z\in Q(t)\). 
        Therefore, \(Q(t)\) contains \(x\) and all \(i+1\) jobs \(s_{m+1},\ldots,s_{m+i+1}\), so it contains at least \(i+2\) jobs from \(\{s_1,\ldots,s_{m+i+1}\}\), including \(s_{m+i+1}\). 
        \item Suppose that \(P(\tau)\) does not contain all of \(s_{m+1},\ldots,s_{m+i}\). 
        Let \(j\in\{1,\ldots,i\}\) be the maximum value such that \(s_{m+j}\notin P(\tau)\). 
        By maximality of \(j\), every job \(s_{m+j+1},\ldots,s_{m+i+1}\) belongs to \(P(\tau)\). 
        By the induction hypothesis for \(j\), there is a post-boundary time slot \(t_j\) such that \(s_{m+j}\in Q(t_j)\) and \(|Q(t_j)\cap\{s_1,\ldots,s_{m+j}\}|\ge j+1\). 
        Since \(s_{m+j}\notin P(\tau)\), the crossing rule forces every job \(s_{m+j+1},\ldots,s_{m+i+1}\) to also belong to \(Q(t_j)\). 
        Therefore, \(Q(t_j)\) contains at least \((j+1)+(i+1-j)=i+2\) jobs from \(\{s_1,\ldots,s_{m+i+1}\}\), including \(s_{m+i+1}\). 
        This gives the required witness for \(i+1\). 
    \end{enumerate}
    This completes the induction. 
    Taking \(i=m\) gives a post-boundary time slot that must contain at least \(m+1\) jobs, which contradicts the capacity \(m\). 
    Therefore, \(q\le 2m-1\). 
\end{proof}

It remains to describe live jobs that are not split. 
Such a job either is rejected or is processed completely on one side of the boundary. 
Canonicality implies that the accepted jobs of this type are separated by a single index cut. 

\begin{lemma}\label{lem:proper_consecutive_fully_scheduled}
    Fix a boundary \(h\) and a canonical schedule. 
    Let \(S_h\) be the set of jobs that are split at \(h\), and let \(L_h = \{j\in J : r_j\le h<d_j\}\) be the jobs with windows that contain (not necessarily active) slots on both sides of the boundary. 
    Define \(\mathcal{X}_h = L_h\setminus S_h\) as the set of jobs with a window on both sides of boundary \(h\), but not split. 
    Then there is an index \(c\) such that every accepted job \(i\in\mathcal{X}_h\) scheduled entirely before or at \(h\) has index \(i\le c\), and every accepted job \(i\in\mathcal{X}_h\) scheduled entirely after \(h\) has index \(i>c\). 
\end{lemma}
\begin{proof}
    Consider the accepted jobs \(j,i\in\mathcal{X}_h\) with \(j<i\), where \(j\) is scheduled entirely after the boundary and \(i\) entirely before the boundary. 
    Let \(t\) be a boundary time slot before the boundary with \(i\in A(t)\) and a boundary slot \(t'>t\) after the boundary with \(j\in A(t')\). 
    Since neither job is split, it holds that \(j\notin A(t)\) and \(i\notin A(t')\). 
    Therefore, \((j,i)\) is an inversion pair, which contradicts canonicality. 
    Any accepted non-split jobs scheduled before the boundary form a prefix, and the jobs scheduled after the boundary form a suffix. 
\end{proof}

The dynamic program maintains this cut index explicitly. 
At each boundary, the DP stores the few jobs that are split across the boundary together with their progress. 
All other live jobs are summarized by the cut index. 
Jobs on the closed side of the cut may not be processed later, while jobs on the open side may still be accepted. 
This means that the DP does not have to remember rejected jobs explicitly. 
Lemma~\ref{lem:proper_few_splits} guarantees that in some optimal canonical, the number of split jobs at every boundary is at most \(2m-1\). 

Since release time and deadline windows can be large, we restrict the time horizon to a smaller number of relevant time slots depending only on \(n\) and \(K\). 
As in Lemma~\ref{lem:canonical_time_slots}, any schedule that uses at most \(K\) active time slots can be transformed into a schedule where every active time slot lies within distance \(K\) of some release time or deadline. 
After this compression and renumbering, the time horizon has \(T=O(nK)\) time slots. 

The dynamic program scans time from left to right. 
Boundary \(h\) is the point immediately after time slot \(h\), for \(h=0,\ldots,T\). 
Recall that \(L_h=\{j\in J:r_j\le h<d_j\}\) is the set of jobs whose windows contain slots on both sides of boundary \(h\). 
At boundary \(h\), the DP state is a tuple \(S,\rho,c,b\). 
Here \(S\subseteq L_h\) is the set of jobs that have been started but not yet completed, and we require \(|S|\le 2m-1\). 
For every \(j\in S\), the value \(\rho(j)\in\{1,\ldots,p_j-1\}\) is the amount of processing already given to \(j\). 
The integer \(b\in\{0,\ldots,K\}\) is the number of active slots used so far. 
Lastly, \(c\in\{0,\ldots,n\}\) is the cut index for live jobs that are not carried in \(S\). 
A live non-carried job \(j\in L_h\setminus S\) with \(j\le c\) is called \emph{closed}, which means it has either already been completed in the prefix, or it has been rejected and may not be processed later. 
A live non-carried job \(j>c\) is still \emph{open} and may be processed in the suffix. 
So the DP does not remember rejected jobs individually, since they are represented by the closed side of the cut. 
The value of a state is the maximum total weight of jobs completed in the prefix among all prefix schedules realizing the state. 
For any schedule that realizes \((S,\rho,c,b)\), the jobs in \(S\) are exactly the jobs with positive unfinished progress after slot \(h\), and every processed job outside \(S\) has already been completed. 
The initial state is \((\emptyset,\emptyset,0,0)\). 

We now describe the transition between boundary \(h-1\) and \(h\). 
Let the previous state be \((S^-,\rho^-,c^-,b^-)\) and \((S^+,\rho^+,c^+,b^+)\) the new state. 
The DP guesses at most \(m\) jobs \(X_h\subseteq\{j\in J: r_j\le h\le d_j\}\) which are scheduled in time slot \(h\). 
The new budget value is \(b^+ = b^- + 1\) if \(X_h\neq\emptyset\) and \(b^+=b^-\) otherwise, and we require that \(b^+\le K\). 
The transition does not process any job that was already closed at the previous boundary. 
Therefore, if \(j\in X_h\setminus S^-\) was live and not carried at boundary \(h-1\), i.e. \(j\in L_{h-1}\setminus S^-\), then we require \(j>c^-\). 
Given the guess \(X_h\), the progress after slot \(h\) is simply updated by \(x_j = \rho^-(j) + 1\) if \(j\in X_h\), and \(x_j = \rho^-(j)\) otherwise.
For any newly started job in \(j\in X_h\setminus S^-\), define \(x_j=1\). 
These are the only jobs that can have positive progress after slot \(h\). 
Let \(C_h=\{j: x_j=p_j\}\) be the set of jobs that complete in slot \(h\) and let \(U_h=\{j:0<x_j<p_j\}\) be the set of unfinished jobs after slot \(h\). 
The unfinished jobs must be carried across the boundary \(h\), so \(S^+=U_h\) and \(\rho^+(j)=x_j\) for all \(j\in S^+\). 
Note that every job in \(U_h\) must belong to \(L_h\), otherwise it has unfinished work without a future time slot in its window which would make the transition infeasible. 
For the cut, any job that completes in slot \(h\) and is still live across boundary \(h\) must be closed immediately. 
So \(j\in C_h\cap L_h\) implies \(j\le c^+\). 
Furthermore, any closed jobs that are still available and non-carried must stay closed, so for every \(j\in L_{h-1}\cap L_h\) with \(j\notin S^-\cup S^+\) and \(j\le c^-\), we require \(j\le c^+\). 
Within these constraints, the DP can choose the cut index \(c^+\) freely. 
Choosing a larger cut simply means rejecting additional open live jobs. 

When all these conditions hold, the transition updates the value by adding the weight of the completed jobs in time slot \(h\). 
\[
\mathrm{DP}_h(S^+,\rho^+,c^+,b^+) = \max\left\{
\mathrm{DP}_h(S^+,\rho^+,c^+,b^+),
\mathrm{DP}_{h-1}(S^-,\rho^-,c^-,b^-) + \sum_{j\in C_h}w_j
\right\}
\]
Therefore, each job contributes weight exactly once, namely when its progress first reaches \(p_j\). 
Note that the DP itself does not enforce canonicality. 
This is not needed, since every DP path describes a feasible schedule. 
Canonicality is only used for completeness, since by Lemma~\ref{lem:proper_exchange}, there is an optimal canonical schedule, and Lemmas~\ref{lem:proper_few_splits} and~\ref{lem:proper_consecutive_fully_scheduled} imply that this schedule is represented by the described states. 

Final states are states at boundary \(T\) with \(S=\emptyset\). 
The optimum value is the maximum DP value among all such final states. 
For each boundary, there are \((nK)^{O(m)}\) choices for \(S\) and \(\rho\), and polynomially many choices for \(c\) and \(b\). 
For each state, the set \(X_h\) has size at most \(m\), so all choices for \(X_h\) can be enumerated in \(n^{O(m)}\) time. 
Since the compressed horizon has \(T=O(nK)\) time slots, the total running time is \((nK)^{O(m)}\).

\begin{theorem}\label{thm:proper_weighted_throughput_xp}
    Weighted throughput maximization with active-time budget \(K\) on proper intervals can be solved in \((nK)^{O(m)}\) time. 
    In particular, the problem is \(\mathrm{XP}\) solvable for parameter \(m\) and polynomial in \(n\) and \(K\).
\end{theorem}
\begin{proof}
    We first prove soundness. 
    Every DP transition only processes jobs whose windows contain the current slot.
    Furthermore, it processes at most \(m\) jobs in that time slot. 
    The budget variable increases exactly when the slot is active, so any final DP path uses at most \(K\) active time slots. 
    Jobs with unfinished progress are stored explicitly in the carried set \(S\), and a job contributes weight to the objective value only when its progress reaches \(p_j\). 
    The cut condition prevents closed live jobs from being processed later. 
    Therefore, every final state with \(S=\emptyset\) corresponds to a feasible schedule whose value is exactly the DP value. 

    Now we prove completeness. 
    By Lemma~\ref{lem:proper_exchange}, there exists an optimal canonical schedule. 
    In such a schedule, Lemma~\ref{lem:proper_few_splits} implies that at most \(2m-1\) jobs are split across any boundary. 
    These are exactly the jobs stored in \(S\), together with their progress values. 
    Lemma~\ref{lem:proper_consecutive_fully_scheduled} implies that all other live jobs can be summarized by one cut index. 
    Closed jobs lie on one side of the cut and open jobs on the other. 
    Therefore, the optimal canonical schedule induces a valid sequence of DP states and transitions. 
    The DP value is therefore at least the optimum, and soundness gives equality. 
    The running time bound follows from the state transition counts described above. 
\end{proof}

\section{An FPT algorithm for laminar intervals with weighted throughput}\label{sec:nested}
We now consider active-time scheduling when the job windows form a laminar family. 
We assume throughout this section that the family of windows is laminar, i.e. for any two jobs \(i,j\in J\), the windows \(W_i\) and \(W_j\) are either disjoint, or one contains the other. 
We discard every job with \(p_j>K\), since such a job cannot be completed using at most \(K\) active slots. 
Throughout this section, we assume \(p_j\le K\) for all jobs. 
We give an FPT algorithm parameterized in \(K\) and \(m\). 
The main tool that we use is that for each nested family, we can construct all possible capacity profiles of some active time slots over each nested interval in FPT time. 

The laminar family is represented by a rooted forest. 
Each node \(u\) of the forest corresponds to one distinct window \(W_u\).
Let \(J(u)=\{j\in J: W_j=W_u\}\) be the jobs whose window is exactly \(W_u\). 
A node \(v\) is a child of \(u\) if \(W_v\subsetneq W_u\) and there is no node \(z\) such that \(W_v\subsetneq W_z\subsetneq W_u\). 
Therefore, the children of \(u\) are the maximal subwindows contained in \(W_u\). 
By laminarity, the children of a node are pairwise disjoint. 
We add a dummy root \(r\), with \(J(r)=\emptyset\), whose children are the roots of the laminar forest. 
The dummy root has no actual window and is only used to combine the connected components. 

For a non-dummy node \(u\), we define its private region as the time slots in \(W_u\) that are not contained in any child window \(P(u) = W_u\setminus\bigcup_{v\text{ child of }u}W_v\). 
Let \(\lambda_u=|P(u)|\) be the number of private time slots. 
These are the time slots that are available to jobs in \(J(u)\) and to ancestors of \(u\), but not to any descendant of \(u\). 
For the dummy root \(r\), we set \(\lambda_r=0\). 

The DP only needs to remember the remaining capacities of active slots inside a subtree. 
Since at most \(K\) time slots can be active and every slot has capacity \(m\), this information can be stored as a small histogram of at most \(m\) buckets filled up to at most \(K\). 
A \emph{capacity profile} is a vector 
\[
\mathbf{h} = (h_0,h_1,\dots,h_m)\in\mathbb{Z}_{\ge 0}^{m+1},
\]
where \(h_c\) is the number of active slots with exactly \(c\) units of remaining capacity. 
Therefore, \(h_0\) is the number of full active time slots, \(h_m\) is the number of active time slots that are currently empty, and \(\sum_{c=0}^m h_c\) is the total number of active time slots represented by the profile. 

We only consider all profiles using at most \(K\) active time slots
\[
\mathcal{H}_{K,m} = \left\{\mathbf{h}\in\mathbb{Z}_{\ge 0}^{m+1} : \sum_{c=0}^{m}h_c\le K\right\}. 
\]
Therefore, the number of profiles that we consider is \(|\mathcal{H}_{K,m}| = \binom{K+m+1}{m+1}\) by a standard stars-and-bars argument with a total sum at most \(K\). 
For a node \(u\) and a profile \(\mathbf{h}\in\mathcal{H}_{K,m}\), define \(\mathrm{DP}_u[\mathbf{h}]\) as the maximum total weight of accepted jobs whose windows are contained in \(W_u\), subject to the condition that after scheduling, these accepted jobs using only slots inside \(W_u\), the active slots in the subtree of \(u\) have residual-capacity profile \(\mathbf{h}\). 
If no such schedule exists, we set \(\mathrm{DP}_u[\mathbf{h}]=-\infty\). 
The important point is that the precise identities of the active slots are irrelevant to ancestors of \(u\). 
Every ancestor job can use every slot inside \(W_u\). 
Therefore, for the purpose of future scheduling decisions, it is enough to know how many active time slots have each remaining capacity. 
In particular, empty active slots are kept in the profile, since they may later be used by jobs corresponding to ancestor of \(u\). 
The histogram also preserves the restriction that a job uses at most one unit in any single slot. 
When a job of size \(p_j\) is scheduled, the transition chooses \(p_j\) distinct slots from the profile and decreases each chosen slot by one unit. 

At node \(u\), before scheduling the jobs in \(J(u)\), we may reserve some slots in the private region \(P(u)\). 
These reserved time slots are potential active slots that can be used by jobs at node \(u\) or by ancestor jobs. 
Since the global budget is \(K\), it is never useful to open more than \(K\) such slots, and of course we cannot reserve more than \(\lambda_u\) private slots. 
For \(a\in\{0,\dots,\min\{K,\lambda_u\}\}\), reserving \(a\) private slots gives the profile \(\mathbf{e}^{(a)} = (0,0,\ldots,0,a)\), where the entry \(a\) is in coordinate \(m\). 
That is, these \(a\) slots each have remaining capacity \(m\). 
If a reserved slot is never used by any job, it can be deleted from the final schedule. 
Counting it in the profile is only a conservative way to enforce that at most \(K\) slots are ever reserved. 

We process the laminar tree bottom-up. 
Suppose the children of \(u\) are \(v_1,\dots,v_q\) with \(q\le n\). 
First we combine the DP tables of the children and the reserved private slots of \(u\). 
Initialize the initial DP table \(M^{(0)}\) by setting \(M^{(0)}[\mathbf{e}^{(a)}]=0\) for all \(a\in\{0,\dots,\min\{K,\lambda_u\}\}\), and setting all other entries to \(-\infty\). 
For each child \(v_i\) with \(i=1,\dots,q\), construct \(M^{(i)}\) from \(M^{(i-1)}\) as follows. 
For every pair of profiles \(\mathbf{a},\mathbf{h}\in\mathcal{H}_{K,m}\) with the previous DP table \(M^{(i-1)}[\mathbf{a}]>-\infty\), and a non-negative child value \(\mathrm{DP}_{v_i}[\mathbf{h}]>-\infty\), and \(\mathbf{s}=\mathbf{a}+\mathbf{h}\in\mathcal{H}_{K,m}\) being the combined valid profile, update 
\[
M^{(i)}[\mathbf{s}] = \max\left\{M^{(i)}[\mathbf{s}], M^{(i-1)}[\mathbf{a}]+\mathrm{DP}_{v_i}[\mathbf{h}]\right\}. 
\]
This DP avoids enumerating all tuples of child profiles explicitly, since \(M^{(i)}\) stores the best value after combining the first \(i\) children. 
After all children have been combined, \(M^{(q)}\) describes the best schedules inside the child subtrees of \(u\), together with any private slots opened in \(P(u)\). 
No job in \(J(u)\) has been scheduled yet. 

We now process the jobs in \(J(u)\) one by one. 
Let the current table be \(T\), initially equal to \(M^{(q)}\). 
For a job \(j\in J(u)\), we construct a new table \(T^{new}\). 
For each job, there are two choices. 
\begin{itemize}
    \item Reject job \(j\). 
    The profile is unchanged and no weight is gained, so we set \(T^{new}[\mathbf{h}]=\max\{T^{new}[\mathbf{h}],T[\mathbf{h}]\}\) for every \(\mathbf{h}\) with \(T[\mathbf{h}]>-\infty\). 
    \item Accept job \(j\). 
    If job \(j\) is accepted, we must place \(p_j\) units of \(j\) into \(p_j\) distinct active slots inside \(W_u\). 
    Let \(\mathbf{h}=(h_0,h_1,\dots,h_m)\) be the current residual-capacity profile. 
    To schedule job \(j\), choose integers \(a_1,\dots,a_m\) where \(a_c\) is the number of slots with residual capacity \(c\) that receive one unit of job \(j\). 
    The choice is feasible if \(0\le a_c\le h_c\) for all \(c\in\{1,\dots,m\}\), and \(\sum_{c=1}^{m}a_c=p_j\). 
    After placing job \(j\), the selected slots each lose one unit of remaining capacity. 
    The resulting profile \(\mathbf{h'}\) is given by \(h'_0=h_0+a_1\), and for \(1\le c\le m\), \(h'_c=h_c-a_c+a_{c+1}\), where \(a_{m+1}=0\). 
    Then we update \(T^{new}[\mathbf{h'}] = \max\left\{T^{new}[\mathbf{h'}], T[\mathbf{h}] + w_j\right\}\).
\end{itemize}

After both cases are applied, we replace \(T\) by \(T^{new}\) and continue with the next job of \(J(u)\). 
Once all jobs of \(J(u)\) are processed, we set \(\mathrm{DP}_u=T\). 
After processing the dummy root \(r\), the optimum value is \(\max_{\mathbf{h}\in\mathcal{H}_{K,m}}\mathrm{DP}_r[\mathbf{h}]\).
Every profile in \(\mathcal{H}_{K,m}\) uses at most \(K\) active slots, so the active-time budget is respected. 

\begin{theorem}
    Weighted-throughput active-time scheduling on laminar intervals can be solved exactly in \(f(K,m)\cdot n^{O(1)}\) time. 
\end{theorem}
\begin{proof}
    We prove correctness and then analyze the running time. 
    We first show that every solution represented by the DP corresponds to a feasible active-time schedule. 
    At each node \(u\), the DP reserves some number of slots in the private region \(P(u)\). 
    It never reserves more than \(\lambda_u=|P(u)|\) such slots, so these slots can be chosen as actual time slots in \(P(u)\). 
    The DP also combines slots already constructed in child subtrees. 
    Since child windows are pairwise disjoint, combining child profiles simply takes the disjoint union of their active time slots. 

    When the DP accepts a job \(j\in J(u)\), it selects \(p_j\) active time slots with positive residual capacity and decreases the residual capacity of each selected slot by one. 
    Therefore no active slot receives more than \(m\) jobs in total. 
    Because the \(p_j\) selected slots are distinct, job \(j\) receives at most one unit of processing in any single time slot. 

    All time slots used for job \(j\) lie inside the subtree of \(u\), and therefore inside \(W_u=W_j\). 
    Therefore, job \(j\) is scheduled only inside its own window. 
    Finally, every accepted job has received exactly \(p_j\) units of processing time, and rejected jobs contribute no weight. 
    Therefore the DP constructs a feasible schedule with a total objective value weight that is exactly the value stored in the DP table. 

    We now show that every feasible schedule is represented by the DP. 
    Consider any feasible schedule using at most \(K\) active slots. 
    We show by induction on the laminar tree that the DP can reproduce its behavior. 
    Fix a node \(u\). 
    The active slots inside the child subtrees of \(u\) are represented by the DP tables of the children by the induction hypothesis. 
    The active slots of the schedule in the private region \(P(u)\) are represented by the initialization step that opens private slots. 
    Since the schedule uses at most \(K\) active time slots globally, the DP allows enough private slots to represent them. 
    Combining the children and the private slots gives exactly the residual-capacity profile that remains before scheduling the jobs in \(J(u)\). 
    Now consider a job \(j\in J(u)\). 
    If \(j\) is not completed in the schedule, the DP can reject it. 
    If \(j\) is completed, then the schedule assigns its \(p_j\) units to \(p_j\) distinct slots inside \(W_u\). 
    At the moment before accounting for \(j\), each of those slots has positive remaining capacity. 
    Therefore, the DP can choose the corresponding vector \((a_1,\dots,a_m)\), where one unit from each of those slots is consumed, and add weight \(w_j\). 
    Processing the jobs of \(J(u)\) in this way reproduces the decision at node \(u\) of the schedule. 
    By induction, the DP can represent every feasible schedule. 
    Hence the value computed at the dummy root is at least the optimum schedule value. 
    Soundness gives the reverse inequality, so the DP computes the optimum weighted throughput. 

    For the running time, the number of capacity profiles is \(|\mathcal{H}_{K,m}| = \binom{K+m+1}{m+1}\). 
    Combining one child table with the current table requires considering pairs of profiles, and therefore takes \(O(|\mathcal{H}_{K,m}|^2)\) time. 
    Over all nodes, the total number of child-combinations is \(O(n)\). 
    For each job \(j\), and for each profile \(\mathbf{h}\), the DP enumerates all feasible job allocation vectors \((a_1,\dots,a_m)\) with \(\sum_{c=1}^{m}a_c=p_j\). 
    The number of such vectors is at most \(A_{K,m}\le (K+1)^m\), since \(p_j\le K\). 
    Thus all job-processing transitions take \(O(n\cdot |\mathcal{H}_{K,m}|\cdot A_{K,m})\) time. 
    Therefore, the total running time is 
    \[
    O\left(n\cdot|\mathcal{H}_{K,m}|^2 + n\cdot |\mathcal{H}_{K,m}|\cdot A_{K,m}\right) = f(K,m)\cdot n^{O(1)}.
    \] 
    This proves that the problem is FPT parameterized by \(K\) and \(m\). 
\end{proof}

\section{Conclusion and Future Work}
We showed that the weighted-throughput version of active-time scheduling is already nontrivial even when capacity is unbounded. 
In the unbounded capacity setting, the problem is NP-hard and admits no FPTAS unless \(\mathrm{P}=\mathrm{NP}\), but it has an \(\Omega(1/\log K)\)-approximation. 
For more structured windows, we gave a pseudo-polynomial \((nK)^{O(m)}\)-time algorithm for proper intervals, and an FPT algorithm for laminar intervals parameterized by \(K\) and \(m\). 

Several questions remain open. 
The most immediate one is whether the \(\Omega(1/\log K)\)-approximation for unbounded capacity can be improved. 
It would also be interesting to understand whether proper or laminar intervals admit polynomial-time algorithms, or fixed-parameter algorithms under smaller parameter sets.

\printbibliography 

@article{DBLP:journals/corr/abs-2112-03255,
  author       = {Sagnik Saha and
                  Manish Purohit},
  title        = {NP-completeness of the Active Time Scheduling Problem},
  journal      = {CoRR},
  volume       = {abs/2112.03255},
  year         = {2021},
  url          = {https://arxiv.org/abs/2112.03255},
  eprinttype    = {arXiv},
  eprint       = {2112.03255},
  bibsource    = {dblp computer science bibliography, https://dblp.org}
}

@article{DBLP:journals/siamcomp/PapadimitriouY79,
  author       = {Christos H. Papadimitriou and
                  Mihalis Yannakakis},
  title        = {Scheduling Interval-Ordered Tasks},
  journal      = {{SIAM} J. Comput.},
  volume       = {8},
  number       = {3},
  pages        = {405--409},
  year         = {1979},
  url          = {https://doi.org/10.1137/0208031},
  doi          = {10.1137/0208031},
  bibsource    = {dblp computer science bibliography, https://dblp.org}
}

@article{DBLP:journals/disopt/KubiakRP09,
  author       = {Wieslaw Kubiak and
                  Djamal Rebaine and
                  Chris N. Potts},
  title        = {Optimality of {HLF} for scheduling divide-and-conquer {UET} task graphs
                  on identical parallel processors},
  journal      = {Discret. Optim.},
  volume       = {6},
  number       = {1},
  pages        = {79--91},
  year         = {2009},
  url          = {https://doi.org/10.1016/j.disopt.2008.09.001},
  doi          = {10.1016/J.DISOPT.2008.09.001},
  bibsource    = {dblp computer science bibliography, https://dblp.org}
}

@article{DBLP:journals/algorithmica/ChangGK14,
  author       = {Jessica Chang and
                  Harold N. Gabow and
                  Samir Khuller},
  title        = {A Model for Minimizing Active Processor Time},
  journal      = {Algorithmica},
  volume       = {70},
  number       = {3},
  pages        = {368--405},
  year         = {2014},
  url          = {https://doi.org/10.1007/s00453-013-9807-y},
  doi          = {10.1007/S00453-013-9807-Y},
  bibsource    = {dblp computer science bibliography, https://dblp.org}
}

@article{DBLP:journals/scheduling/ChangKM17,
  author       = {Jessica Chang and
                  Samir Khuller and
                  Koyel Mukherjee},
  title        = {{LP} rounding and combinatorial algorithms for minimizing active and
                  busy time},
  journal      = {J. Sched.},
  volume       = {20},
  number       = {6},
  pages        = {657--680},
  year         = {2017},
  url          = {https://doi.org/10.1007/s10951-017-0531-3},
  doi          = {10.1007/S10951-017-0531-3},
  bibsource    = {dblp computer science bibliography, https://dblp.org}
}

@inproceedings{DBLP:conf/isaac/CaoFLMRU22,
  author       = {Nairen Cao and
                  Jeremy T. Fineman and
                  Shi Li and
                  Juli{\'{a}}n Mestre and
                  Katina Russell and
                  Seeun William Umboh},
  editor       = {Sang Won Bae and
                  Heejin Park},
  title        = {Nested Active-Time Scheduling},
  booktitle    = {33rd International Symposium on Algorithms and Computation, {ISAAC}
                  2022, December 19-21, 2022, Seoul, Korea},
  series       = {LIPIcs},
  volume       = {248},
  pages        = {36:1--36:16},
  publisher    = {Schloss Dagstuhl - Leibniz-Zentrum f{\"{u}}r Informatik},
  year         = {2022},
  url          = {https://doi.org/10.4230/LIPIcs.ISAAC.2022.36},
  doi          = {10.4230/LIPICS.ISAAC.2022.36},
  bibsource    = {dblp computer science bibliography, https://dblp.org}
}

@inproceedings{DBLP:conf/spaa/KumarK18,
  author       = {Saurabh Kumar and
                  Samir Khuller},
  editor       = {Christian Scheideler and
                  Jeremy T. Fineman},
  title        = {Brief Announcement: {A} Greedy 2 Approximation for the Active Time
                  Problem},
  booktitle    = {Proceedings of the 30th on Symposium on Parallelism in Algorithms
                  and Architectures, {SPAA} 2018, Vienna, Austria, July 16-18, 2018},
  pages        = {347--349},
  publisher    = {{ACM}},
  year         = {2018},
  url          = {https://doi.org/10.1145/3210377.3210659},
  doi          = {10.1145/3210377.3210659},
  bibsource    = {dblp computer science bibliography, https://dblp.org}
}

@article{DBLP:journals/scheduling/CalinescuW21,
  author       = {Gruia C{\u{a}}linescu and
                  Kai Wang},
  title        = {A new {LP} rounding algorithm for the active time problem},
  journal      = {J. Sched.},
  volume       = {24},
  number       = {5},
  pages        = {543--552},
  year         = {2021},
  url          = {https://doi.org/10.1007/s10951-020-00676-1},
  doi          = {10.1007/S10951-020-00676-1},
  bibsource    = {dblp computer science bibliography, https://dblp.org}
}

@article{DBLP:journals/cacm/Albers10,
  author       = {Susanne Albers},
  title        = {Energy-efficient algorithms},
  journal      = {Commun. {ACM}},
  volume       = {53},
  number       = {5},
  pages        = {86--96},
  year         = {2010},
  url          = {https://doi.org/10.1145/1735223.1735245},
  doi          = {10.1145/1735223.1735245},
  bibsource    = {dblp computer science bibliography, https://dblp.org}
}

@inproceedings{DBLP:conf/birthday/ChauL20,
  author       = {Vincent Chau and
                  Minming Li},
  editor       = {Ding{-}Zhu Du and
                  Jie Wang},
  title        = {Active and Busy Time Scheduling Problem: {A} Survey},
  booktitle    = {Complexity and Approximation - In Memory of Ker-I Ko},
  series       = {Lecture Notes in Computer Science},
  volume       = {12000},
  pages        = {219--229},
  publisher    = {Springer},
  year         = {2020},
  address={Cham},
  url          = {https://doi.org/10.1007/978-3-030-41672-0\_13},
  doi          = {10.1007/978-3-030-41672-0\_13},
  bibsource    = {dblp computer science bibliography, https://dblp.org}
}

@article{DBLP:journals/scheduling/DemaineGHSZ13,
  author       = {Erik D. Demaine and
                  Mohammad Ghodsi and
                  MohammadTaghi Hajiaghayi and
                  Amin S. Sayedi{-}Roshkhar and
                  Morteza Zadimoghaddam},
  title        = {Scheduling to minimize gaps and power consumption},
  journal      = {J. Sched.},
  volume       = {16},
  number       = {2},
  pages        = {151--160},
  year         = {2013},
  doi          = {10.1007/S10951-012-0309-6},
  bibsource    = {dblp computer science bibliography, https://dblp.org}
}

@inproceedings{DBLP:conf/soda/WinklerZ03,
author = {Winkler, Peter and Zhang, Lisa},
title = {Wavelength assignment and generalized interval graph coloring},
year = {2003},
isbn = {0898715385},
publisher = {Society for Industrial and Applied Mathematics},
address = {USA},
booktitle = {Proceedings of the Fourteenth Annual ACM-SIAM Symposium on Discrete Algorithms},
pages = {830–831},
numpages = {2},
location = {Baltimore, Maryland},
series = {SODA '03}
}

@article{DBLP:journals/tcs/ShalomVWYZ14,
  author       = {Mordechai Shalom and
                  Ariella Voloshin and
                  Prudence W. H. Wong and
                  Fencol C. C. Yung and
                  Shmuel Zaks},
  title        = {Online optimization of busy time on parallel machines},
  journal      = {Theor. Comput. Sci.},
  volume       = {560},
  pages        = {190--206},
  year         = {2014},
  url          = {https://doi.org/10.1016/j.tcs.2014.07.017},
  doi          = {10.1016/J.TCS.2014.07.017},
  bibsource    = {dblp computer science bibliography, https://dblp.org}
}

@inproceedings{DBLP:conf/wads/ChauFLWZ019,
  author       = {Vincent Chau and
                  Shengzhong Feng and
                  Minming Li and
                  Yinling Wang and
                  Guochuan Zhang and
                  Yong Zhang},
  editor       = {Zachary Friggstad and
                  J{\"{o}}rg{-}R{\"{u}}diger Sack and
                  Mohammad R. Salavatipour},
  title        = {Weighted Throughput Maximization with Calibrations},
  booktitle    = {Algorithms and Data Structures - 16th International Symposium, {WADS}
                  2019, Edmonton, AB, Canada, August 5-7, 2019, Proceedings},
  series       = {Lecture Notes in Computer Science},
  volume       = {11646},
  pages        = {311--324},
  publisher    = {Springer},
  year         = {2019},
  url          = {https://doi.org/10.1007/978-3-030-24766-9\_23},
  doi          = {10.1007/978-3-030-24766-9\_23},
  bibsource    = {dblp computer science bibliography, https://dblp.org}
}

@article{DBLP:journals/disopt/LampisKM11,
  author       = {Michael Lampis and
                  Georgia Kaouri and
                  Valia Mitsou},
  title        = {On the algorithmic effectiveness of digraph decompositions and complexity
                  measures},
  journal      = {Discret. Optim.},
  volume       = {8},
  number       = {1},
  pages        = {129--138},
  year         = {2011},
  url          = {https://doi.org/10.1016/j.disopt.2010.03.010},
  doi          = {10.1016/J.DISOPT.2010.03.010},
  bibsource    = {dblp computer science bibliography, https://dblp.org}
}

\appendix

\section{Interval-ordered chains and NP-hardness}\label{sec:hardness}
The input for IOC consists of \(n\) unit-time jobs partitioned into chains \(S_1,\dots, S_\ell\). 
Within each chain \(S_i\), the jobs are totally ordered. 
In addition, the chains themselves are partially ordered by an interval order. 
That is, each chain \(S_i\) is represented by a closed interval \(I(S_i)\) on the real line, and \((S_i,S_{i'})\) is a precedence constraint exactly when \(\sup I(S_i) < \inf I(S_{i'})\). 
Such a chain-level precedence constraint applies to all jobs in the two chains, i.e. if \((S_i,S_{i'})\) is present, then every job of \(S_i\) precedes every job of \(S_{i'}\). 
This model (IOC) can be denoted in three-field notation as \(P\mid ioc, p_j=1\mid C_{\max}\). 

The IOC model strictly generalizes both interval orders and chains, since singleton chains recover interval-ordered unit jobs, while parallel chains recover the chain setting. 
Both interval orders and chains (and the more general divide-and-conquer precedence graphs) are polynomial time solvable using simple greedy procedures \cite{DBLP:journals/siamcomp/PapadimitriouY79,DBLP:journals/disopt/KubiakRP09}. 
The interesting part here is that, using a similar construction as Saha and Purohit for the active-time scheduling hardness \cite{DBLP:journals/corr/abs-2112-03255}, we can show that IOC is in fact NP-hard, which we show in Appendix~\ref{sec:hardness}. 
IOC is a relaxation of active-time scheduling because active-time scheduling enforces strict release times and deadlines. 
Allowing arbitrarily long overlap between the release-time and deadline windows recovers the IOC model. 
Our paper focuses on the more general active-time scheduling model, but we believe the IOC model might be of independent interest. 
We prove that makespan minimization for interval-ordered chains with unit-time jobs (\(P|ioc,p_j=1|C_{max}\)) is NP-hard. 
Our construction closely follows the time slot layout and gadget ideas of the NP-hardness proof shown by Saha and Purohit for the active time scheduling problem \cite{DBLP:journals/corr/abs-2112-03255}. 
However, we need to encode these gadgets as interval-ordered unit chains, which does not directly follow from the active time scheduling instance. 
We refer many times to the construction of Saha and Purohit, which all refer to \cite{DBLP:journals/corr/abs-2112-03255} in this section. 
All gadgets we define are derived from the gadgets by Saha and Purohit. 
Following Saha and Purohit, we use a reduction from Balanced SAT. 

\begin{definition}[Balanced SAT, \cite{DBLP:journals/corr/abs-2112-03255}]
    Given is a boolean formula \(F(x_1,x_2,\dots,x_n)\) where the number of variables \(n\) is even. 
    \(F\) is presented in CNF (AND of ORs) and contains \(k\) clauses \(C_1,\dots,C_k\), where each clause \(C_i\) is an OR of several literals. 
    The problem asks whether there is a satisfiable assignment of the boolean variables that sets exactly \(n/2\) variables to \emph{true}. 
\end{definition}

We reduce an instance of Balanced SAT to the decision variant of an instance of \(P|ioc,p_j=1|C_{max}\) with a threshold \(\tau\) such that the SAT formula is satisfiable and balanced if and only if there is a feasible schedule of makespan at most \(\tau\). 
We use the set of time slots introduced by Saha and Purohit for the active time scheduling problem for our precedence constraint setting and give a self-contained reduction from Balanced SAT to the decision version of scheduling interval-ordered chains. 

Let \(F\) be a balanced SAT formula on \(n\) variables \(x_1,\dots,x_n\) and clauses \(C_1,\dots,C_k\). 
Let \(n_i\) be the number of literals in clause \(C_i\) and let \(n_x\) denote the number of occurrences of variable \(x\) in formula \(F\). 
The target makespan is set to \(\tau := k + 2 + \sum^k_{i=1}n_i\). 
The number of machines is set to \(m := 2n+2\)
We define the set of unit-length time slots along a line the same way as Saha and Purohit \cite{DBLP:journals/corr/abs-2112-03255}. 
\begin{itemize}
    \item Define a single left boundary slot \(L\). 
    \item For each clause \(C_i\) in order \(1,\dots,k\), first define two time slots corresponding to the negated literal occurrences of \(C_i\), then a clause time slot \(C_i\), followed by the time slots for the positive literal occurrences of \(C_i\). 
    \item Finally, define the right boundary time slot \(R\). 
\end{itemize}
An overview of an example reduction is shown in Figure~\ref{fig:NP_typical_example}. 
We define three gadgets: a \emph{variable} gadget, a \emph{clause} gadget, and a \emph{copy} gadget. 
The interval-ordered precedence constraints are constructed by the intervals in the real line. 
The intervals span several time slots, which represent different parts of the reduction. 
\begin{figure}
    \centering
    \includegraphics[page=3]{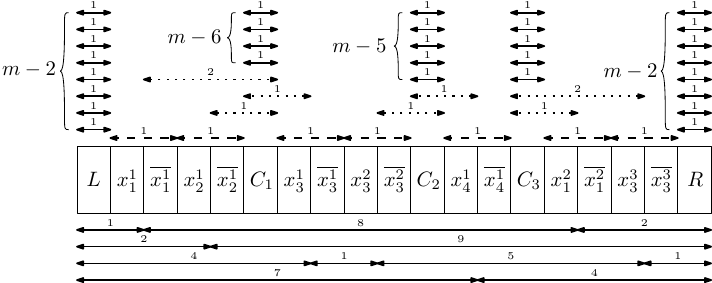}
    \caption{A typical translation of a balanced SAT instance with \(n=4\) variables \(F=(\overline{x_1}\vee\overline{x_2}\vee x_3)\wedge(\overline{x_3}\vee x_4)\wedge(x_1\vee x_3)\). 
    This example is the same as the Saha and Purohit example \cite{DBLP:journals/corr/abs-2112-03255}. 
    The numbers indicate the length of the chain in the set. 
    The dashed lines represent the \emph{variable} gadgets, the dotted lines represent the \emph{clause} gadgets, and the lines below the boxes represent the \emph{copy} gadgets.  
    }
    \label{fig:NP_typical_example}
\end{figure}

We define four families of sets, denoted \(X,Y,Z,B\), which are variable, clause, copy, and filler gadgets, respectively. 
In the variable gadget, we consider each literal-occurrence and define two consecutive slots. 
For each occurrence \(\ell\) of \(x_i\) or \(\overline{x_i}\) in clause \(C_i\), we create two consecutive slots placed on the appropriate side of \(C_i\) according to the sign of \(\ell\). 
Let \(j\) index the number of occurrences of \(x_i\) if \(F\). 

\begin{definition}[Variable gadget]
    The \emph{variable gadget} for variable occurrence \(x^j_i\) is represented by set \(X^{(j)}_{x_i}\) with chain length \(|X^{(j)}_{x_i}|=1\), and interval \(I(X^{(j)}_{x_i}) = [x^j_i,\overline{x^j_i}]\). 
    The gadget is depicted in Figure~\ref{fig:NP_variable_gadget}. 
\end{definition}
\begin{figure}
    \centering
    \includegraphics[page=4]{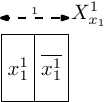}
    \caption{Variable gadget}
    \label{fig:NP_variable_gadget}
\end{figure}

Intuitively, the single job of \(X^{(j)}_{x_i}\) must be scheduled in one of the two slots, which represents the truth value of \(x\). 
Each clause gadget \(C_i\) is constructed from variable gadgets which each represent a literal in the clause. 

\begin{definition}[Clause gadget]
    The \emph{clause gadget} represents clause \(C_k\) with \(n_k\) literals. 
    An example is shown in Figure~\ref{fig:NP_clause_gadget}. 
    For each occurrence \(x^j_i\in C_k\), there is a corresponding \(X^j_{x_i}\) of the corresponding variable gadget that spans the two literal slots. 
    Consider the clause time slot \(C_k\). 
    For a literal occurrence \(x^j_i\), if it is negated, then the corresponding variable gadget is placed to the left of \(C_k\) (and if it is not negated, the corresponding variable gadget is placed to the right of \(C_k\)). 
    To control capacity at the clause time slot \(C_k\), add \(m-n-n_k+1\) filler sets \(B^1_{C_k},\dots,B^{m-n-n_k+1}_{C_k}\) with a chain of length \(1\) with the interval in exactly \(C_k\). 
    For each literal \(x^j_i\) in clause \(C_k\), an interval \(Y^j_{x_i}\) of length one plus the number of variable gadgets it spans is created. 
    \begin{enumerate}
        \item If the literal \(x^j_i\) is negated in clause \(C_k\), \(Y^j_{x_i}\) is the interval \([\overline{x^j_i},C_k]\). 
        \item If the literal \(x^j_i\) is not negated in clause \(C_k\), \(Y^j_{x_i}\) is the interval \([C_k,x^j_i]\). 
    \end{enumerate}
\end{definition}
\begin{figure}
    \centering
    \includegraphics[page=5]{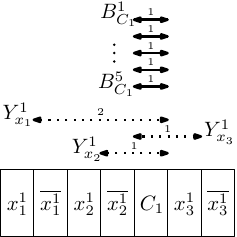}
    \caption{Clause gadget}
    \label{fig:NP_clause_gadget}
\end{figure}

We add boundary filler sets at the endpoints \(L,R\) of the construction to enforce the balance constraint. 
At time slot \(L\) (resp. \(R\)), add \(m-n/2\) sets \(B^1_L,\dots,B^{m-n/2}_L\) (resp. \(B^1_R,\dots,B^{m-n/2}_R\)) each with a chain of length \(1\) with interval \([L,L]\) (resp. \([R,R]\)). 
We now define the only sets that overlap \(L\) or \(R\) other than the filler sets, which are the copy gadget sets. 

\begin{definition}[Copy gadget]
    The \emph{copy gadget} ensures that all variables in the clauses are set to the same value. 
    An example copy gadget is depicted in Figure~\ref{fig:NP_copy_gadget}. 
    For each variable \(x_i\), the following copy gadget sets are created. 
    \begin{enumerate}
        \item \(Z^1_{x_i}\) spans \([L,x^1_i]\)
        \item \(Z^j_{x_i}\) spans \([\overline{x^{j-1}_i},x^{j}_i]\) for all \(2\leq j\leq n_{x_i}\)
        \item \(Z^{n_{x_i}+1}_{x_i}\) spans \([\overline{x^{n_{x_i}}_{x_i}},R]\)
    \end{enumerate}
    The length of the chain of set \(Z^j_{x_i}\) is one plus the number of variable gadget intervals \(X\) and the number of \(C_i\) time slots that are fully contained in the interval of set \(Z^j_{x_i}\). 
\end{definition}
\begin{figure}
    \centering
    \includegraphics[page=6]{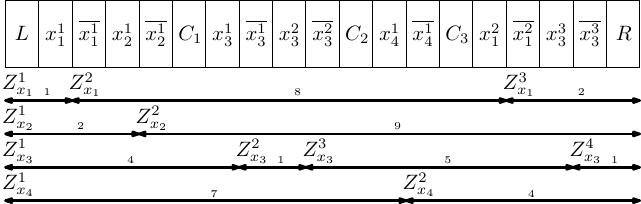}
    \caption{Copy gadget}
    \label{fig:NP_copy_gadget}
\end{figure}

Using these gadgets, we construct an instance of interval-ordered chains. 
For an optimal schedule, we look at the critical path of the constructed instance. 
We observe the following lemma by the construction of the intervals. 

\begin{lemma}\label{lem:NP_critical_path}
    The constructed instance has a critical path of length \(\tau = k+2+\sum^k_{i=1} n_i\). 
\end{lemma}
\begin{proof}
    We first show that the critical path is at least \(k+2+\sum^k_{i=1} n_i\). 
    We consider the critical path induced by precedence-ordered sets. 
    Consider the sets \(B^1_L\), the sets \(X^j_{x_i}\) for all \(j\) and \(i\), sets \(B^1_{C_1},\dots,B^1_{C_k}\) and finally set \(B^1_R\). 
    All considered sets are pairwise disjoint, and therefore, the chain of each set contributes fully to the critical path. 
    The length is therefore at least \(1+\sum^k_{i=1}n_i+k+1 = k+2+\sum^k_{i=1} n_i = \tau\). 
    
    For the upper bound, we consider any precedence constrained sets. 
    The number of unit jobs that are part of any chain \(Z\) is by definition equal to the number of variable gadgets \(X\) it spans, plus one. 
    Since each chain \(Z\) partially overlaps at least one variable gadget, there exists a critical chain that does not include any set \(Z\). 
    For each \(Y\), this is similarly defined as the number of variable gadgets \(X\) it fully contains plus one, and it partially overlaps at least one variable gadget. 
    Therefore, there exists a critical chain that does not include any set \(Y\). 
    This leaves exactly the sets \(L,X,B,R\). 
    Similarly to the lower bound, the length is therefore at most \(1+\sum^k_{i=1}n_i+k+1 = k+2+\sum^k_{i=1} n_i = \tau\). 
    Combining the lower and upper bound, gives that the constructed instance has a critical path of length \(\tau = k+2+\sum^k_{i=1} n_i\)
\end{proof}

We show how the construction enforces the truth value of a literal occurrence. 

\begin{lemma}\label{lem:NP_choice_precedence_constraints}
    Consider literal \(x^j_i\), which is the \(j^{th}\) occurrence of variable \(x_i\). 
    If literal \(x^j_i\) is negated, then \(Z^j_{x_i}\prec Y^j_{x_i}\). 
    If literal \(x^j_i\) is not negated, then \(Y^j_{x_i}\prec Z^{j+1}_{x_i}\). 
    Moreover, variable gadget \(X^j_{x_i}\) overlaps both respective sets. 
\end{lemma}
\begin{proof}
By construction. 
\end{proof}

To show that the reduction holds, we show that the optimal makespan to the scheduling instance is \(\tau\) if and only if the Balanced SAT formula \(F\) is solvable. 

\begin{lemma}\label{lem:np_balancedSAT_to_schedule}
    If the Balanced SAT instance \(F\) is satisfiable with exactly \(n/2\) variables set to true, then the interval-ordered chains instance admits a feasible schedule of makespan \(\tau\). 
\end{lemma}
\begin{proof}
    Consider a balanced satisfying assignment for \(F\). 
    We describe a schedule of makespan \(\tau\) that respects all precedence and capacity constraints. 
    We assign jobs from sets to specific time slots. 

    At time slot \(L\), schedule all \(m-n/2\) sets \(B_L\). 
    For each variable \(x_i\) that is set to \(false\), schedule one job of \(Z^1_{x_i}\) in \(L\). 
    This uses \((m-n/2)+(n/2)=m\) capacity. 
    At \(R\), schedule all \(m-n/2\) sets \(B_R\) and for each variable \(x_i\) that is set to \(true\), assign the last job in the chain of set \(Z^{n_{x_i}+1}_{x_i}\) in \(R\). 
    This uses \((m-n/2)+(n/2)=m\) capacity. 

    For the literals in each clause, if \(x^j_i\) is satisfied, schedule one job of \(Y^j_{x_i}\) with \(X^j_{x_i}\). 
    If \(x^j_i\) is not satisfied, schedule one job of \(Z^j_{x_i}\) with \(X^j_{x_i}\) if \(x^j_i\) is negated, or \(Z^{j+1}_{x_i}\) with \(X^j_{x_i}\) if \(x^j_i\) is not negated. 
    For each clause \(C_k\), we schedule all \(m-n-n_k+1\) fillers \(B_{C_k}\) at \(C_k\). 
    Schedule a job from all \(Y^j_{x_i}\) for which \(x_i\) is not satisfied with \(C_k\). 
    Since the clause is satisfied, there is at least one literal that is satisfied, say \(x^j_i\), for which one job from \(Y^j_{x_i}\) is scheduled with \(X^j_{x_i}\). 
    This gives a remaining capacity of \(m-(m-n-n_k+1)-(n_k-1)=n\) at slot \(C_k\). 
    Schedule a job from all \(Z\) jobs that overlap with \(C_k\). 
    This uses exactly \(n\) capacity which saturates the capacity at slot \(C_k\). 

    All remaining jobs of \(Y\) and \(Z\) sets can now be placed inside their intervals with the scheduled jobs since \(m=2n+2\) and each literal pair contributes two unit jobs at each of their slots. 

    All jobs respect the precedence constraints, as only incomparable sets share the same time slot. 
    Furthermore, as shown in Lemma~\ref{lem:NP_critical_path} all slots used lies within the \(\tau\) slots. 
    Therefore, the interval-ordered chains instance admits a feasible schedule of makespan \(\tau\). 
\end{proof}

\begin{lemma}\label{lem:np_schedule_to_balancedSAT}
    If the interval-ordered chains instance admits a feasible schedule of makespan \(\tau\), then the Balanced SAT instance \(F\) is satisfiable with exactly \(n/2\) variables set to true. 
\end{lemma}
\begin{proof}
    Consider an optimal schedule \(S^*\) with makespan \(\tau\). 
    Since by Lemma~\ref{lem:NP_critical_path}, \(\tau\) equals the critical path lower bound, the optimal schedule \(S^*\) cannot schedule beyond this critical path. 

    
    First, we verify the copy gadget. 
    We say that an occurrence of variable \(x^j_i\) is set to \(true\) if the job unit from \(X^j_{x_i}\) schedules with \(Z^j_{x_i}\), and set to \(false\) if the job unit from \(X^j_{x_i}\) schedules with \(Z^{j+1}_{x_i}\). 
    Consider the case that for some set \(Z^j_{x_i}\), no job unit from its chain is scheduled with either \(X^{j-1}_{x_i}\) or \(X^j_{x_i}\). 
    By construction, the chain of \(Z^j_{x_i}\) is one more than the remaining number of variable gadgets and clause time slots that its interval spans. 
    Therefore, this case would add a time slot to the makespan that is not part of the critical path, which gives a contradiction. 
    %
    For all the sets with \(Z^j_{x_i}\) with \(j\in[2,n_{x_i}]\), this gives the same value. 
    In the case that \(Z^2_{x_i}\) is scheduled with \(X^1_{x_i}\), a job unit from \(Z^1_{x_i}\) must schedule with \(B_L\). 
    Otherwise, a time slot that is not part of the critical path would be added to the makespan. 
    Since \(x_i\) is set to \(false\) if \(Z^2_{x_i}\) schedules with \(X^1_{x_i}\) and the remaining \(Z^j_{x_i}\) copies the value, \(B_L\) must schedule with a job unit from \(Z^1_{x_i}\) when \(x_i\) is set to \(false\). 
    A similar argument holds for the need to schedule with \(B_R\) when \(x_i\) is set to \(true\). 
    %
    Since exactly \(n/2\) processors are available to schedule with \(B_L\) and \(B_R\), this means that exactly \(n/2\) variables must be set to \(false\), and similarly exactly \(n/2\) variables must be set to \(true\). 
    Since we verified that the copy gadget ensures all occurrences of the variable are set to the same value, we will talk about the value of a variable. 
    
    We now look at the clauses and check all are satisfied. 
    By the construction of the copy gadget sets must schedule with \(B_{C_i}\) for all \(C_i\) in order to keep scheduling with the critical path. 
    This uses \(n\) processors for all \(C_i\). 
    
    %
    Assume that some clause \(C_k\) is not satisfied. 
    Since \(C_k\) is not satisfied, all negated variables are set to \(true\) and all non-negated variables are set to \(false\). 
    By the construction of the copy gadget that we verified earlier, if a negated variable \(x^j_i\) is set to \(true\), set \(Z^j_{x_i}\) has a job unit that schedules with \(X^j_{x_i}\). 
    By the interval precedence constraints, no job from the set \(Y^j_{x_i}\) may schedule with a job unit from set \(X^j_{x_i}\). 
    Similarly for a non-negated variable \(x^j_i\) that is set to \(false\), set \(Z^{j+1}_{x_i}\) has a job unit that schedules with \(X^j_{x_i}\) and no job from the set \(Y^j_{x_i}\) may schedule with a job unit from set \(X^j_{x_i}\). 
    By the construction of the sets \(Y^j_{x_i}\), the chain has length equal to the number of \(X\) intervals it fully contains, plus one. 
    Since \(Y^j_{x_i}\) can't schedule with \(X^j_{x_i}\), the additional job unit must schedule with \(B_{C_k}\). 
    Since all copy gadget jobs must schedule in \(C_k\), and there are \(m-n-n_k+1\) sets \(B_{C_k}\) that schedule in parallel that are part of the critical path and \(n_k\) job units from the sets \(Y^j_{x_i}\) schedule in parallel, this exceeds the capacity \(m\) which gives a contradiction. 
    Therefore, there must exist one variable set such that clause \(C_k\) is satisfied. 
\end{proof}

Combining Lemma~\ref{lem:np_balancedSAT_to_schedule} and Lemma~\ref{lem:np_schedule_to_balancedSAT}, this proves both directions of the reduction. 
The NP-hardness follows for the complexity of scheduling interval-ordered chains.

\end{document}